\documentclass[hidelinks, 11 pt, a4paper]{article}
\usepackage{amsfonts,amsmath,amssymb,amsthm}
\usepackage{appendix}
\usepackage{subfigure}
\usepackage[pdftex]{graphicx}
\usepackage{natbib}
\usepackage[hypertexnames=false]{hyperref}
\usepackage{setspace}
\usepackage{pdfsync}
\usepackage{lscape}
\usepackage{verbatim}
\usepackage{enumerate}
\usepackage[hmargin=1.25in,vmargin=1.25in]{geometry}

\usepackage{booktabs, multirow, rotating, threeparttable} 

\hypersetup{colorlinks, linkcolor = {teal}, citecolor = {blue}, urlcolor ={blue!80!black}}

\makeatletter
\renewcommand\paragraph{\@startsection{paragraph}{4}{\z@}%
            {-2.5ex\@plus -1ex \@minus -.25ex}%
            {1.25ex \@plus .25ex}%
            {\normalfont\normalsize\bfseries}}
\makeatother
\usepackage{mathrsfs,dsfont}
\usepackage[dvipsnames,svgnames, x11names]{xcolor}
\usepackage{pgfplots}
\usetikzlibrary{patterns}
\usepackage{caption}

\DeclareGraphicsExtensions{.jpg}

\def\qed{\rule{2mm}{2mm}}
\def\indep{\perp \!\!\!\!\perp}

\newtheorem{lemma}{Lemma}[section]

\newtheorem{assumption}{Assumption}[section]
\newtheorem{algorithm}{Algorithm}[section]
\newtheorem{remark}{Remark}[section]

\let\oldmarginpar\marginpar
\renewcommand{\marginpar}[2][rectangle,draw,fill=black, text=white,text width= 2cm,rounded corners]{
    \oldmarginpar{
    \tiny \tikz \node at (0,0) [#1]{#2};}
    }

\begin{document}
\author{
Yong Cai \\ 
Department of Economics \\ 
Northwestern University \\ 
\url{yongcai2023@u.northwestern.edu}
\and
Ivan A.\ Canay \\ 
Department of Economics \\ 
Northwestern University \\ 
\url{iacanay@northwestern.edu}
\and
Deborah Kim \\ 
Department of Economics \\ 
Northwestern University \\ 
\url{deborahkim@u.northwestern.edu}
\and
Azeem M.\ Shaikh\\
Department of Economics\\
University of Chicago \\
\url{amshaikh@uchicago.edu}
}

\bigskip

\title{On the Implementation of Approximate Randomization Tests in Linear Models with a Small Number of Clusters\thanks{We would like to thank Matthew Thomas for excellent research assistance developing the \texttt{R} and \texttt{Stata} packages for this paper. The research of the fourth author is supported by NSF Grant SES-1530661.}}

\maketitle


\begin{spacing}{1.1}
\begin{abstract}
This paper provides a user's guide to the general theory of approximate randomization tests developed in \cite{canay/etal:17} when specialized to linear regressions with clustered data. An important feature of the methodology is that it applies to settings in which the number of clusters is small -- even as small as five.  We provide a step-by-step algorithmic description of how to implement the test and construct confidence intervals for the parameter of interest. In doing so, we additionally present three novel results concerning the methodology: we show that the method admits an equivalent implementation based on weighted scores; we show the test and confidence intervals are invariant to whether the test statistic is  studentized or not; and we prove convexity of the confidence intervals for scalar parameters. We also articulate the main requirements underlying the test, emphasizing in particular common pitfalls that researchers may encounter.  Finally, we illustrate the use of the methodology with two applications that further illuminate these points: one to a linear regression with clustered data based on \cite{mengQY2015} and a second to a linear regression with temporally dependent data based on \cite{munyo2015first}. The companion {\tt R} and {\tt Stata} packages facilitate the implementation of the methodology and the replication of the empirical exercises. 
\end{abstract}
\end{spacing}



\thispagestyle{empty} 
\newpage
\setcounter{page}{1}

\section{Introduction}


This paper provides a user's guide to the general theory of approximate randomization tests (ARTs) developed in \cite{canay/etal:17} when specialized to linear regressions with clustered data.  Here, clustered data refers to data that may be grouped so that there may be dependence within each cluster, but distinct clusters are approximately independent in a way to be made precise below.  Such data is remarkably common, including not only data that are naturally grouped into clusters, such as villages or repeated observations over time on individual units, but also data with weak temporal dependence, in which pseudo-clusters may be formed using blocks of consecutive observations.  An important feature of the methodology is that it applies to commonly encountered settings in which the number of clusters is small -- even as small as five.  In this respect, the proposed methodology contrasts sharply and meaningfully with many commonly employed methods for inference in such settings.  We briefly elaborate on this point in our discussion of related literature below.

A principal\textbf{} goal of this paper is to make the general theory developed in \cite{canay/etal:17} more accessible by providing a step-by-step algorithmic description of how to implement the test and construct confidence intervals for the quantity of interest in these types of settings. In order to do so, we develop three novel results concerning the methodology in Section \ref{sec:results}. Our first result shows that what we view as the most natural implementation of the test, as described in Algorithm \ref{algo:crs1}, is numerically equivalent to an alternative implementation based on weighted scores (see Algorithm \ref{algo:crs2}). Our second result shows that when the parameter of interest is a scalar parameter, studentizing or not the $t$-statistic entering the test does not affect the results of the test or the associated confidence intervals. We therefore focus on the unstudentized statistic in Algorithm \ref{algo:crs1}. Finally, our third result shows that the confidence sets for scalar parameters that are conceptually described by test inversion are indeed a closed interval of the real line. This further leads to a simple closed form expression for the lower and upper bound of the confidence intervals (see Algorithm \ref{algo:crs2}). These results are new to this paper and play an important role in developing simple algorithms for the implementation of ARTs.

We additionally provide a discussion of the main requirements underlying the test in Section \ref{sec:requirements}.  These requirements essentially demand that the quantity of interest is suitably estimable cluster-by-cluster.  As discussed further in Section \ref{sec:requirements}, when this is not satisfied, a researcher need not conclude that it is not possible to exploit the results in \cite{canay/etal:17}.  Instead, several remedies are possible, including clustering more coarsely or changing the specification to ensure that this requirement is satisfied.  We provide two applications that further elucidate these points: one to a linear regression with clustered data based on \cite{mengQY2015} and a second to a linear regression with temporally dependent data based on \cite{munyo2015first}.  The required software to replicate these empirical exercises and to aid researchers wishing to employ the methods elsewhere is provided in both {\tt R} and {\tt Stata}.\footnote{The {\tt Stata} and {\tt R} packages {\tt ARTs} can be downloaded from \hyperlink{http://sites.northwestern. edu/iac879/software/}{http://sites.northwestern. edu/iac879/software/}.}

The methodology described in this paper is part of a large and active literature on inference with clustered data.  Following \cite{bertrand2004much}, researchers are acutely aware of the need to adjust inferences appropriately to account for this sort of dependence.  Many of the most commonly employed methods for doing so, however, are inadequate for the unusually common situation in which the number of clusters is small.  Conventional wisdom suggests that the number of clusters is small when it is less than forty.  For example, the method described in \cite{liang1986longitudinal}, which has enjoyed considerable popularity due to its availability in software packages such as {\tt Stata}, is widely acknowledged to perform poorly when this rule-of-thumb is not satisfied.  Similarly, the cluster wild bootstrap described in \cite{cameron2008bootstrap} requires either a sufficiently large number of clusters or, as shown by \cite{canay/santos/shaikh:20}, stringent homogeneity across clusters, to perform reliably.  As explained further in Section \ref{sec:requirements}, the methods developed in \cite{canay/etal:17} and described in this paper, require neither a large number of clusters nor such homogeneity across clusters.  We note that the methods by  \cite{ibragimov2010t,ibragimov/muller:16}, which are closely related to the ones described here, also do not require such restrictions, but are generally less powerful and permit testing a less rich variety of hypotheses.  See \cite{canay/etal:17} for further discussion of these points as well as \cite{conley2018inference} for an insightful and thorough review of the related literature more broadly.

The remainder of this paper is organized as follows.  In Section \ref{sec:setup}, we first formalize the setting and establish some notation.  We then describe the implementation of approximate randomization tests (ARTs) in an algorithmic fashion, including how to use these tests to construct confidence intervals for the quantity of interest. In Section \ref{sec:results} we present three results that play an important role in developing these algorithms. In Section \ref{sec:requirements}, we articulate the main requirements underlying the tests and discuss remedies for cases where these requirements are not satisfied.  Our two empirical applications are contained in Section \ref{sec:applications}.  Finally, we provide some concluding remarks in Section \ref{sec:conclude}.

\section{Review of ARTs in regression models}\label{sec:setup}
We start by reviewing the inference approach proposed by \cite{canay/etal:17} in the context of a linear regression model with clustered data. In order to do so, we index clusters by $j \in J \equiv \{1, \ldots, q\}$ and units in the $j$th cluster by $i \in I_{n,j} \equiv \{1, \ldots, n_j\}$. We also denote by $n=\sum_{j=1}^q n_j$ the total number of observations. The observed data consists of an outcome of interest, $Y_{i,j}$, and a vector of covariates, $Z_{i,j} \in \mathbf R^{d_z}$, that are related through the equation
\begin{equation}\label{eq:main}
    Y_{i,j} = Z_{i,j}'\beta + \epsilon_{i,j}~,
\end{equation}
where $\beta \in \mathbf{R}^{d_{z}}$ are unknown parameters and our requirements on $\epsilon_{i,j}$ are explained below in Section \ref{sec:requirements}. Our goal is to test
\begin{equation}\label{eq:H0-1}
    H_0: c^\prime \beta = \lambda \quad \text{vs.}\quad  H_1: c^\prime \beta \neq \lambda ~,
\end{equation}
for given values of $c\in \mathbf R^{d_z}$ and $\lambda \in \mathbf R$, at level $\alpha \in (0,1)$. An important special case of this framework is a test of the null hypothesis that a particular component of $\beta$ equals a given value, i.e., 
\begin{equation*}
    H_0: \beta_{\ell} = \lambda \quad \text{vs.}\quad  H_1: \beta_{\ell} \neq \lambda ~,
\end{equation*}
for some $\ell \in \{1,\dots,d_{z}\}$, by simply setting $c$ to be a standard unit vector with a one in the $\ell$th component and zeros otherwise. More generally, the approach we describe below extends immediately to the case where the hypothesis of interest involves multiple elements of $\beta$, in which case the test becomes
 \begin{equation}\label{eq:H0-2}
    H_0: R\beta = \Lambda \quad \text{vs.}\quad  H_1: R\beta \neq \Lambda ~,
\end{equation}
for a given $p\times d_z$-dimensional matrix $R$ and $p$-dimensional vector $\Lambda$, at level $\alpha \in (0,1)$.  

ARTs were developed more generally in \cite{canay/etal:17} and admit a variety of different applications that go beyond the linear model considered here. For example, the method accommodates non-linear models, non-linear hypotheses, or even applications that go beyond inference with a small number of clusters (e.g., \cite{canay/kamat:18} develop a variation that applies to inference in the regression discontinuity design). Here, we abstract away from the generality of the method and focus on the steps needed to use ARTs to test the null hypothesis in \eqref{eq:H0-1} in the context of the model in \eqref{eq:main}. 

\subsection{How to implement ARTs}\label{sec:algorithm}
The most straightforward way to test the hypotheses in \eqref{eq:H0-1} via ARTs is by following the steps described in Algorithm \ref{algo:crs1} below. 

\begin{algorithm}[ARTs via within-cluster estimates]\label{algo:crs1}
This implementation of ARTs involves the following steps:
\begin{enumerate}
    \setlength{\itemindent}{-1em}
    \item[] {\bf Step 1}: For each cluster $j\in J$, run an ordinary least squares regression of $Y_{i,j}$ on $Z_{i,j}$ using the $n_j$ observations in cluster $j$. Denote the corresponding estimators of $\beta$ by $$\{\hat \beta_{n,j}:j \in J\}~.$$ 
    \item[] {\bf Step 2}: For each $j\in J$, define the random variables 
    \begin{equation}\label{eq:Snj}
        S_{n,j} \equiv \sqrt{n_j}(c'\hat\beta_{n,j}-\lambda)~,
    \end{equation} 
    and then construct the test statistic 
    \begin{equation}\label{eq:statistic}
        T_n = \Big|\frac{1}{q}\sum_{j=1}^q S_{n,j} \Big|~.
    \end{equation}
    \item[] {\bf Step 3}: Let $\mathbf G = \{1,-1\}^q$, so $g=(g_1,\dots,g_q)\in\mathbf G$ is simply a $q$-dimensional vector with elements $g_j$ being either $1$ or $-1$. For any element $g\in\mathbf G$, define  
    \begin{equation}\label{eq:g-statistic}
        T_n(g) =\Big| \frac{1}{q}\sum_{j=1}^q g_jS_{n,j} \Big|~.
    \end{equation}
    \item[] {\bf Step 4}: Compute the $1-\alpha$ quantile of $\{T_n(g): g\in\mathbf G\}$ as
    \begin{equation}\label{eq:c-hat}
        \hat c_n(1-\alpha) \equiv \inf \left \{ u \in \mathbf R: \frac{1}{|\mathbf G|}\sum_{g\in \mathbf G} I\{T_n(g) \leq u\} \geq 1-\alpha \right \}~.
    \end{equation}
    \item[] {\bf Step 5}: Compute the test as 
    \begin{equation}\label{eq:art-test}
        \phi_n\equiv I\{T_n > \hat c_n(1-\alpha)\}~,
    \end{equation}
    where $T_n$ is as in \eqref{eq:statistic} and $\hat c_n(1-\alpha)$ is as in \eqref{eq:c-hat}. The associated $p$-value is
    \begin{equation}\label{eq:art-pvalue}
        \hat p_n \equiv \frac{1}{|\mathbf G|}\sum_{g\in \mathbf G} I\{T_n(g) \ge T_n\}~,
    \end{equation}
    where $T_n(g)$ is as in \eqref{eq:g-statistic}.
\end{enumerate}
\end{algorithm}

Algorithm \ref{algo:crs1} involves five steps that are easy to implement from a computational standpoint, but some of the steps deserve some clarification. Step 1 involves $q$ within-cluster regressions that lead to $q$ estimates of $\beta$. This essentially demands that the parameter $\beta$ is identified cluster-by-cluster, and may fail to hold if some of the variables in the vector $Z_{i,j}$ are constant within cluster. We discuss possible remedies for this problem in Section \ref{sec:requirements} and illustrate their use in one of the applications in Section \ref{sec:applications}. An important feature of the method is that from Step 2 onwards, the original data is no longer needed as all the calculations only involve the $q$ estimators of the parameter $\beta$ obtained in Step 1. 

Step 2 defines a type of unstudentized $t$-statistic that is appropriate for the null hypothesis in \eqref{eq:H0-1}. We discuss the connection to its studentized version in Section \ref{sec:studentized} below. If the null hypothesis of interest is the one in \eqref{eq:H0-2}, then a Wald-type test statistic could be used instead, i.e., 
\begin{equation}\label{eq:T-Wald}
    T_n^{\rm wald} \equiv q\Big(\frac{1}{q}\sum_{j=1}^q S_{n,j}\Big)'\Sigma_{S}^{-1}\Big(\frac{1}{q}\sum_{j=1}^q S_{n,j} \Big)  ~,
\end{equation}   
where
\begin{equation*}\label{eq:Snj-Sigma}
    S_{n,j} \equiv \sqrt{n}(R\hat\beta_{n,j}-\Lambda)\quad \text{ and }\quad  \Sigma_{S} \equiv \frac{1}{q}\sum_{j=1}^q S_{n,j}S_{n,j}^{\prime}~.
\end{equation*}   

Step 3 does not require one to recompute the estimates of $\beta$. It rather uses the $q$ estimates from Step 1 and applies sign changes to the $q$-dimensional vector $\{S_{n,j}:j \in J\}$. Since the cardinality of $\mathbf G$ is $|\mathbf G|=2^q$, it exceeds $2000$ when $q>10$ and in such cases it may be convenient to use a stochastic approximation. This may be done while still controlling the rejection probability under the null hypothesis \citep[see][Remark 2.2]{canay/etal:17}. Formally, in this case we let 
\begin{equation}\label{eq:G-hat}
    \hat{\mathbf G} \equiv \{g^1,\dots, g^B\}~,
\end{equation} 
where $g^1=\iota\equiv (1,\dots,1)$ is the identity vector and $g^{b}=(g^{b}_1,\dots,g^{b}_q)$, for $b=2,\dots,B$, are i.i.d.\ Rademacher random variables; i.e., each $g^{b}_j$ equals $\pm1$ with equal probability. To retain validity of the test regardless of the value of $B$, we require that $g^1=\iota$. We note, however, that the power of the test may still depend on $B$. For this reason, we implement Algorithm \ref{algo:crs1} with $\hat{\mathbf G}$ replacing $\mathbf G$ everywhere and set $B=1000$ (or any other reasonably large number chosen by the analyst). 

Step 4 requires computing the $1-\alpha$ quantile of $\{T_n(g): g\in\mathbf G\}$, which can be typically  obtained by sorting the values of $\{T_n(g): g\in\mathbf G\}$ and then taking the $\lceil |\mathbf G|(1-\alpha) \rceil^{\rm th}$ highest element in the ordered list. Thus, if we denote the ordered values of $\{T_n(g): g\in\mathbf G\}$ by 
$$ T^{(1)}_n\le T^{(2)}_n \le \cdots \le T^{(B)}_n~,$$
then we may define $\hat c_n(1-\alpha)$ in \eqref{eq:c-hat} as $\hat c_n(1-\alpha)= T^{(\lceil |\mathbf G|(1-\alpha) \rceil)}$. This representation suggests that the test may have trivial power for very low values of $q$. For example, when $\alpha=10\%$, this problem arises if $q\le 4$. For $q=5$ the test already has non-trivial power and is only slightly conservative under the null. Similarly, when $\alpha=5\%$ the test has non-trivial power for any $q\ge 6$. 

Step 5 is straightforward and it provides both the test $\phi_n$ and the $p$-value $\hat p_n$. Each of these correspond to the non-randomized version of ARTs as opposed to their randomized counterparts \citep[see Remark 2.4 in][]{canay/etal:17} since practitioners often prefer tests that do not involve exogenous randomness. In any case, the differences between the randomized and non-randomized versions of the test have been found to be minimal in simulations \citep[see, e.g.,][]{canay/etal:17}.

\subsection{How to compute confidence intervals}\label{sec:ci}
We now discuss how to compute confidence intervals for the parameter $c^\prime \beta$ by developing a novel algorithm that exploits the properties derived in Section \ref{sec:convexity}. As before, a particularly important case is when $c$ selects the $\ell$th component of $\beta$ and then the confidence set is simply a confidence interval for $\beta_{\ell}$. Conceptually we can simply form the confidence set by collecting all values of $c^\prime \beta$ that cannot be rejected by our test at level $\alpha$. That is, for the test $\phi_n$ in \eqref{eq:art-test} we define
\begin{equation}\label{eq:tilde-Cn}
    C_n = \{\lambda \in \mathbf{R}: \phi_n=0 \text{ when testing }H_0: c^\prime \beta = \lambda \}~.
\end{equation}
In an asymptotic framework where $n\to \infty$ while $q$ remains fixed, \cite{canay/etal:17} show that $\phi_n$ is asymptotically level $\alpha$ under $H_0$. It follows from that result that, by construction, $C_n$ covers $c^\prime \beta$ with probability at least equal to $1-\alpha$ asymptotically. In Section \ref{sec:convexity} we show that $C_n$ is indeed a closed interval in $\mathbf{R}$ and so it takes the form 
\begin{equation}\label{eq:Cn}
    C_n = [\lambda_{l}, \lambda_{u}]~,
\end{equation}
where $\lambda_{l}$ is the \emph{smallest} value of $\lambda$ that cannot be rejected by $\phi_n$ and $\lambda_{u}$ is the \emph{largest} value of $\lambda$ that cannot be rejected by $\phi_n$. The analysis in Section \ref{sec:convexity} also reveals that $\lambda_{l}$ and $\lambda_{u}$ admit simple closed-form representations that we exploit to develop Algorithm \ref{algo:crs3} below. 

\begin{algorithm}[ART-based confidence intervals for $c^\prime\beta$]\label{algo:crs3}
For $\{\hat \beta_{n,j}:j\in J\}$ as defined in Step 1 of Algorithm \ref{algo:crs1}, the construction of the confidence interval involves the following steps: 
\begin{enumerate}
    \setlength{\itemindent}{-1em}
    \item[] {\bf Step 1}: For every $g\in \mathbf G$, compute the following objects,
    \begin{equation}\label{eq:a-b}
    a(g) \equiv \frac{1}{q} \sum_{j=1}^q \sqrt{n_j}g_j~, \quad b(g) \equiv \frac{1}{q} \sum_{j=1}^q \sqrt{n_j}g_jc'\hat\beta_{n,j}~, \quad \text{and}\quad \lambda_0 \equiv \frac{b(\iota)}{a(\iota)}~,     
    \end{equation} 
    where $\iota = (1,\dots,1)\in \mathbf G$ is the vector with all ones. 
    \item[] {\bf Step 2}: For every $g\in \mathbf G$  define 
    \begin{equation}\label{eq:lambda-l-g}
    \lambda_l(g)\equiv  \begin{cases}
            \frac{b(\iota)}{a(\iota)}\frac{|a(\iota)|}{|a(\iota)| + |a(g)|}  + \frac{b(g)}{a(g)}\frac{|a(g)|}{|a(\iota)| + |a(g)|} &  \text{ if }  \frac{b(g)}{a(g)}\leq \lambda_0 \text{ and } |a(g)|\ne 0 \\[10pt]
            \frac{b(\iota)}{a(\iota)}\frac{|a(\iota)|}{|a(\iota)| - |a(g)|}  - \frac{b(g)}{a(g)}\frac{|a(g)|}{|a(\iota)| - |a(g)|} & \text{ if }  \frac{b(g)}{a(g)}> \lambda_0 \text{ and } |a(g)|\ne 0 \\[10pt]
            \frac{b(\iota)}{a(\iota)} - \frac{|b(g)|}{a(\iota)}& \text{ if } |a(g)|=0 \\[10pt]
            -\infty  & \text{ if } g=\pm \iota
    \end{cases} ~.
    \end{equation}
    and 
    \begin{equation}\label{eq:lambda-u-g}
    \lambda_u(g)\equiv \begin{cases}
            \frac{b(\iota)}{a(\iota)}\frac{|a(\iota)|}{|a(\iota)| + |a(g)|}  + \frac{b(g)}{a(g)}\frac{|a(g)|}{|a(\iota)| + |a(g)|} &  \mbox{ if } \frac{b(g)}{a(g)} \geq \lambda_0 \text{ and } |a(g)|\ne 0 \\[10pt]
            \frac{b(\iota)}{a(\iota)}\frac{|a(\iota)|}{|a(\iota)| - |a(g)|}  - \frac{b(g)}{a(g)}\frac{|a(g)|}{|a(\iota)| - |a(g)|} &  \mbox{ if } \frac{b(g)}{a(g)} < \lambda_0 \text{ and } |a(g)|\ne 0 \\[10pt]
            \frac{b(\iota)}{a(\iota)} + \frac{|b(g)|}{a(\iota)} &  \mbox{ if } |a(g)|= 0 \\[10pt]
            +\infty & \text{ if } g=\pm \iota
    \end{cases}~.
    \end{equation}
    \item[] {\bf Step 3}: Compute the lower bound $\lambda_l$ in the confidence interval \eqref{eq:Cn} as the $\alpha$ quantile of $\left\{ \lambda_l(g) \, : \, g \in \mathbf{G}  \right\}$, i.e., 
    \begin{equation}\label{eq:lambda-l}
        \lambda_l \equiv \inf \left \{ u \in \mathbf R: \frac{1}{|\mathbf G|}\sum_{g\in \mathbf G} I\{\lambda_l(g) \leq u\} \geq \alpha \right \}~.
    \end{equation}
    Compute the upper bound $\lambda_u$ in the confidence interval \eqref{eq:Cn} as the negative of the $\alpha$ quantile of $\left\{ -\lambda_u(g) \, : \, g \in \mathbf{G}  \right\}$, i.e., 
    \begin{equation}\label{eq:lambda-u}
        \lambda_u \equiv -\inf \left \{ u \in \mathbf R: \frac{1}{|\mathbf G|}\sum_{g\in \mathbf G} I\{-\lambda_u(g) \leq u\} \geq \alpha \right \}~.
    \end{equation}
    Report the confidence interval $C_n$ as in \eqref{eq:Cn}.
\end{enumerate}
\end{algorithm} 

Algorithm \ref{algo:crs3} requires three steps that are straightforward to compute and that exploit the results in Section \ref{sec:convexity}. We refer the reader to that section for the details on why $\lambda_l$ and $\lambda_u$ admit the expressions in \eqref{eq:lambda-l} and \eqref{eq:lambda-u}, respectively. 

\section{Three results on implementation of ARTs}\label{sec:results}
Before we review the main requirement underlying ARTs, we present three properties related to the implementation of ARTs that we believe practitioners should be aware of and that are novel to this paper. The first property establishes a connection between the implementation of ARTs as described in Algorithm \ref{algo:crs1} and an alternative implementation based on weighted scores. The second property establishes the numerical equivalence of ARTs for the null in \eqref{eq:H0-1} when the test statistics in \eqref{eq:statistic} is replaced by its studentized version. The third and final result shows that ARTs confidence set for $c'\beta$ is indeed a closed interval in $\mathbf{R}$ and provides a representation for the upper and lower bounds of the interval that lead to Algorithm \ref{algo:crs3}. 

\subsection{Equivalence with weighted scores}
It turns out that ARTs can be implemented by an algorithm that does not involve estimating the parameter $\beta$ within each cluster. This alternative algorithm involves replacing Steps 1 and 2 in Algorithm \ref{algo:crs1} by the two alternative steps described in Algorithm \ref{algo:crs2} below, while keeping Steps 3 to 5 unaffected. 

\begin{algorithm}[ARTs via within-cluster weighted scores]\label{algo:crs2}
This implementation of ARTs involves the following steps:
\begin{itemize}
    \setlength{\itemindent}{-1em}
    \item[] {\bf Step 1$^\prime$}: Run a full-sample least squares regression of $Y_{i,j}$ on $Z_{i,j}$ subject to the restriction imposed by the null hypothesis, i.e., $c'\beta=\lambda$. Denote by $\hat\epsilon^{\rm r}_{i,j}$ the restricted residuals from this regression and by $\hat \beta^{\rm r}_n$ the restricted LS estimator of $\beta$. 
    \item[] {\bf Step 2$^\prime$}: For each cluster $j\in J$, define 
     \begin{equation}\label{eq:Snj-scores}
         S_{n,j}\equiv c'\hat\Omega_{n,j}^{-1}  \frac{1}{\sqrt{n_j}}\sum_{i\in I_{n,j}} Z_{i,j}\hat \epsilon^{\rm r}_{i,j}~,
     \end{equation}
     where
     \begin{equation}\label{eq:Omega-hat-j}
        \hat\Omega_{n,j}\equiv \frac{1}{n_j\textbf{}}\sum_{i\in I_{n,j}} Z_{i,j}Z_{i,j}'
     \end{equation}
    is a $d_z\times d_z$ matrix that is assumed to be full rank with inverse $\hat\Omega_{n,j}^{-1}$.
    \item[]{\bf Steps 3-5}: Same as in Algorithm \ref{algo:crs1}.
\end{itemize}
\end{algorithm}
Note that Steps 3-5 remain unchanged given the alternative definition of $S_{n,j}$ in Step 2$^\prime$. When it comes to Steps 1 and 2, there are two differences worth discussing. The first difference is that Step 1$^\prime$ requires a single full-sample restricted least squares estimator of $\beta$ as opposed to the $q$ cluster-by-cluster estimators in Step 1 of Algorithm \ref{algo:crs1}. The second difference is that Step 2$^\prime$ is based on within-cluster weighted scores as opposed to the centered within-cluster estimates of $\beta$ in Step 2 of Algorithm \ref{algo:crs1}. Interestingly, these two implementations are numerically equivalent and so implementing ARTs via Algorithm \ref{algo:crs1} or Algorithm \ref{algo:crs2} leads to identical results. To see this formally, it is enough to show that $S_{n,j}$ as defined in \eqref{eq:Snj} and \eqref{eq:Snj-scores} are the same using the following argument. For each $j\in J$, 
\begin{align*}
     S_{n,j} &\equiv c'\hat\Omega_{n,j}^{-1}  \frac{1}{\sqrt{n_j\textbf{}}}\sum_{i\in I_{n,j}} Z_{i,j}\hat \epsilon^{\rm r}_{i,j}\\
         &= c'\hat\Omega_{n,j}^{-1}  \frac{1}{\sqrt{n_j}}\sum_{i\in I_{n,j}} Z_{i,j}(Y_{i,j}- Z_{i,j}'\hat \beta^{\rm r}_n )\\
         &= c'\hat\Omega_{n,j}^{-1}  \frac{1}{\sqrt{n_j}}\sum_{i\in I_{n,j}} Z_{i,j}Y_{i,j}-c'\hat\Omega_{n,j}^{-1}  \frac{1}{\sqrt{n_j}}\sum_{i\in I_{n,j}} Z_{i,j}Z_{i,j}'\hat \beta^{\rm r}_n\\
         & = \sqrt{n_j}(c'\hat \beta_{n,j}-c'\beta) - \sqrt{n_j}(c'\hat \beta^{\rm r}_n-c'\beta)\\
         & = \sqrt{n_j}(c'\hat \beta_{n,j}-\lambda)~,
\end{align*}
where the fourth equality follows by adding and subtracting $\sqrt{n_j}c'\beta$ and the last equality holds because $c'\hat \beta_{n}^{\rm r}=c'\beta=\lambda$ under the null hypothesis in \eqref{eq:H0-1}. It thus follows that $S_{n,j}$ in \eqref{eq:Snj} and in \eqref{eq:Snj-scores} are identical and so ARTs can be alternatively implemented via Algorithm \ref{algo:crs1} or \ref{algo:crs2}.  The following lemma summarizes our discussion above:

\begin{lemma}
Let $\hat\Omega_{n,j}$ in \eqref{eq:Omega-hat-j} be full rank for each $j\in J$. Denote by $C_n$ a confidence interval for $c'\beta$ computed using Algorithm \ref{algo:crs1} and by $C_n^\prime$ a confidence interval for $c'\beta$ computed using Algorithm \ref{algo:crs2}.  Then $C_n = C_n^\prime$.
\end{lemma}

\subsection{Equivalence with studentized version of the \emph{t}-statistic}\label{sec:studentized}
The ART defined in \eqref{eq:art-test} of Algorithm \ref{algo:crs1} is based on the unstudentized test statistic $T_n$ defined in \eqref{eq:statistic}. It may perhaps appear more desirable to instead consider the studentized version of this test statistic as studentization commonly improves performance in a variety of other settings. Here, we prove that this is not the case for ARTs when the null hypothesis is the one in \eqref{eq:H0-1} and that both versions of the test statistic lead to numerically identical results. 

To see this, start by defining the studentized version of the test statistic in \eqref{eq:statistic} as $T_n^{\rm s} \equiv T_n^{\rm s}(\iota)$, where for each $g\in \mathbf G$,
\begin{equation}\label{eq:studentized}
    T_n^{\rm s}(g) \equiv \sqrt{q}\frac{\Big|\frac{1}{q}\sum_{j=1}^q g_jS_{n,j} \Big|}{\hat\sigma_{\rm s}(g)} \quad \text{ and }\quad \hat\sigma_{\rm s}(g) \equiv \sqrt{\frac{1}{q}\sum_{j=1}^q \Big(g_jS_{n,j}-\frac{1}{q}\sum_{j=1}^q g_jS_{n,j}\Big)^2}~. 
\end{equation}
Then note that 
\begin{align*}
    \hat{\sigma}_{\rm s}^2(g) & = \frac{1}{q} \sum_{j=1}^q g_j^2S^2_{n,j} - \left(\frac{1}{q}\sum_{j=1}^q g_jS_{n,j}\right)^2 =V_n - T_n^2(g)~,
\end{align*}
where $V_n \equiv \frac{1}{q} \sum_{j=1}^q g_j^2S^2_{n,j}$ does not depend on $g$ as $g_j^2=1$ for all $j\in J$. It follows that we can write the studentized test statistic as
\begin{align*}
   T_n^{\rm s}(g) = \sqrt{q}\frac{T_n(g)}{\sqrt{V_n - T_n^2(g) }}~.
\end{align*}
Since the function $x \mapsto \frac{x}{\sqrt{1 - x^2}}$ is strictly increasing for $x \in [0,1)$, it follows that $T_n^{\rm s}(g)$ is a strictly monotonic transformation of $T_n(g)$ for each $g\in \mathbf G$. We conclude that $I\{T_n(g)\ge T_n(\iota)\}=I\{T_n^{\rm s}(g)\ge T_n^{\rm s}(\iota)\}$ for all $g\in \mathbf G$ and so the ART based on $T_n(g)$ and $T_n^{\rm s}(g)$ are identical.  This discussion is summarized in the following lemma:

\begin{lemma}
Let $\hat\Omega_{n,j}$ in \eqref{eq:Omega-hat-j} be full rank for each $j\in J$. Denote by $C_n$ a confidence interval for $c'\beta$ computed using Algorithm \ref{algo:crs1} and by $C_n^\prime$ a confidence interval for $c'\beta$ computed using Algorithm \ref{algo:crs1} with $T_n^{\rm s}$ in place of $T_n$ and $T_n^{\rm s}(g)$ in place of $T_n(g)$.  Here, $T_n^{\rm s}(g)$ is given by \eqref{eq:studentized} and $T_n^{\rm s}$ is understood to be $T_n^{\rm s}(\iota)$, where $\iota$ is the identity transformation.  Then, $C_n = C_n^\prime$.
\end{lemma}

\subsection{Convexity of the confidence intervals}\label{sec:convexity}
The ART-based confidence intervals for $c'\beta$ defined in \eqref{eq:tilde-Cn} can be computed by test inversion. From a computational standpoint, however, computing confidence sets by test inversion may be cumbersome and the resulting set may not even be an interval. That is, it may not be closed and convex. In this section we prove that this is not a concern for ART-based confidence intervals for $c'\beta$ and so such confidence intervals could be easily computed by a standard bisection algorithm. In fact, our results go even further. We derive closed form expressions for the lower and upper bounds of the confidence interval that imply that computing ART-based confidence intervals for $c'\beta$ is straightforward from a computational standpoint. In order to derive these results, we slightly change our notation to make explicit the dependence on $\lambda$ of each of the elements entering the test in \eqref{eq:art-test}. To this end, let 
\begin{equation*}
    T_n(g,\lambda) \equiv \Big| \frac{1}{q}\sum_{j=1}^q g_j S_{n,j}(\lambda) \Big| \quad \text{ where } \quad S_{n,j}(\lambda) = \sqrt{n_j}(c'\hat\beta_{n,j}-\lambda)~,
\end{equation*}
and note that $T_n = T_n(\iota, \lambda)$. Using this notation, we can re-write the confidence interval in \eqref{eq:tilde-Cn} as
\begin{equation*}
    C_n = \left\{ \lambda \in \mathbf{R} \, : \frac{1}{|\mathbf{G}|}\sum_{g \in \mathbf{G}} I \left\{T_n(g, \lambda)\ge T_n(\iota,\lambda)   \right\} \ge \alpha  \right\}~,
\end{equation*}
which is simply the values of $\lambda$ for which the $p$-value of the test, as defined in \eqref{eq:art-pvalue}, is not below $\alpha$. In order to show that this confidence set is a closed interval, we claim that the $p$-value 
\begin{equation}\label{eq:p-value-ver2}
    \hat p_n(\lambda) = \frac{1}{|\mathbf{G}|}\sum_{g \in \mathbf{G}} I \left\{T_n(g, \lambda)\ge T_n(\iota,\lambda)   \right\}
\end{equation}
is equal to $1$ for $\lambda_0 \equiv b(\iota)/a(\iota)$, monotonically increasing for any $\lambda<\lambda_0$, and monotonically decreasing for any $\lambda>\lambda_0$. The next lemma formalizes this result. 

\begin{lemma}\label{lem:p-value}
Let $\hat\Omega_{n,j}$ in \eqref{eq:Omega-hat-j} be full rank for each $j\in J$. Let $a(g)$,  $b(g)$, and $\lambda_0$ be defined as in \eqref{eq:a-b}. The p-value in \eqref{eq:p-value-ver2} equals
\begin{equation}\label{eq:pvalue-alternative}
    \hat p_n(\lambda) = \begin{cases}
    \frac{1}{|\mathbf{G}|}\sum_{g \in \mathbf{G}} I \left\{\lambda\ge \lambda_{l}(g)\right\} & \text{ for } \lambda<\lambda_0\\[3pt]
    1 & \text{ for } \lambda=\lambda_0\\[3pt]
   \frac{1}{|\mathbf{G}|}\sum_{g \in \mathbf{G}} I \left\{\lambda\le \lambda_{u}(g)\right\}& \text{ for } \lambda>\lambda_0
    \end{cases}~, 
\end{equation}
where $\{\lambda_u(g):g\in \mathbf G\}$ and $\{\lambda_l(g):g\in \mathbf G\}$ are defined in Algorithm \ref{algo:crs3}.
\end{lemma}

\begin{proof}
It is useful to re-write $T_n(g,\lambda)$ in terms of $a(g)$ and $b(g)$. To this end, note that 
\begin{align}
    T_n(g,\lambda) &\equiv \Big| \frac{1}{q}\sum_{j=1}^q g_j S_{n,j}(\lambda) \Big|\notag = \Big| \frac{1}{q}\sum_{j=1}^q g_j \sqrt{n_j}c'\hat\beta_{n,j}-\lambda\frac{1}{q}\sum_{j=1}^q g_j \sqrt{n_j} \Big|\notag \\
    &= |b(g) - \lambda a(g)|~. \label{eq:statistic-ab}
\end{align}
Given $g\in \mathbf G$ and $a(g)\ne 0$, $T_n(g,\lambda)$ is a ``V-shaped'' function of $\lambda$ taking the value $0$ at $\frac{b(g)}{a(g)}$ and with slope $-|a(g)|$ for all $\lambda<\frac{b(g)}{a(g)}$ and slope $|a(g)|\leq a(\iota) $ for all $\lambda>\frac{b(g)}{a(g)}$. Figure \ref{figure:intersections} illustrates this for three values of $g$.

First, note that $I \left\{T_n(g, \lambda_0)\ge T_n(\iota,\lambda_0)   \right\}=I \{T_n(g, \lambda_0)\ge 0 \}=1$ for all $g\in\mathbf G$ and so it follows immediately that $\hat p_n(\lambda_0)=1$.

Second, restrict attention to the set $\Lambda^{+}\equiv \{\lambda \in \mathbf R: \lambda> \lambda_0\}$ where $T_n(\iota,\lambda)$ is linearly increasing. In order to prove that $\hat p_n(\lambda)$ takes the form in \eqref{eq:pvalue-alternative} we prove that $I \left\{T_n(g, \lambda)\ge T_n(\iota,\lambda)   \right\}=I \left\{\lambda\le \lambda_{u}(g)\right\}$ for each $g\in\mathbf G$ by dividing the argument into three cases.

\underline{Case 1}: Consider $g\in \mathbf G$ such that $a(g)\ne 0$ and $|a(g)|\ne a(\iota)$. Since $|a(g)|< a(\iota)$, it follows that $T_n(g, \lambda)$ and $T_n(\iota, \lambda)$ intersect only once on $\Lambda^{+}$ and this holds regardless of whether $\frac{b(g)}{a(g)}< \lambda_0$ or $\frac{b(g)}{a(g)}\ge \lambda_0$ (see Figure \ref{figure:intersections} for a graphical illustration of each of these cases). Denote the intersection point by $\lambda_{u}(g)$ and note that $T_n(g,\lambda)\ge T_n(\iota,\lambda)$ for all $\lambda_0< \lambda\le \lambda_{u}(g)$ and $T_n(g,\lambda)<T_n(\iota,\lambda)$ for all $\lambda>\lambda_u(g)$. Conclude that on $\Lambda^{+}$,
\begin{equation}\label{eq:equiv-lu}
    I\{T_n(g,\lambda)\ge T_n(\iota,\lambda)\} = I\{\lambda\le \lambda_{u}(g)\}~.
\end{equation}
Simple algebra shows that the intersection point $\lambda_u(g)$ takes the form in \eqref{eq:lambda-u-g}.

\underline{Case 2}: Consider $g\in \mathbf G$ such that $a(g) = 0$. Note that $b(\iota) - \lambda a(\iota)<0$ for $\lambda\in\Lambda^{+}$. It thus follows that for $\lambda\in \Lambda^{+}$,
\begin{equation*}
    I\{T_n(g,\lambda)\ge T_n(\iota, \lambda) \} = I\{|b(g)|\ge |b(\iota) - \lambda a(\iota)| \} = I\left\{ \lambda \le \frac{b(\iota)}{a(\iota)} + \frac{|b(g)|}{a(\iota)} \right\}~,
\end{equation*}
and so \eqref{eq:equiv-lu} holds in this case with $\lambda_{u}(g) = \frac{b(\iota)}{a(\iota)} + \frac{|b(g)|}{a(\iota)} $, as defined in \eqref{eq:lambda-u-g}. 


\underline{Case 3}: Consider $g\in \mathbf G$ such that $|a(g)|=a(\iota)$ and so $g=\pm \iota$. If $g=\iota$, $I\{T_n(g,\lambda)\ge T_n(\iota, \lambda) \}=1$ for all $\lambda\in \mathbf R$. We conclude that \eqref{eq:equiv-lu} holds with $\lambda_u(g)=\infty$. If $g=-\iota$, then we have that $a(-\iota)=-a(\iota)$ and $b(-\iota)=-b(\iota)$ so that $\frac{b(-\iota)}{a(-\iota)} =\lambda_0$ and again $I\left\{ T_n(-\iota, \lambda) \ge T_n(\iota, \lambda) \right\} = 1$ for all $\lambda\in \mathbf R$. We conclude that \eqref{eq:equiv-lu} holds with $\lambda_u(g)=\infty$, as defined in \eqref{eq:lambda-u-g}. This completes the proof of \eqref{eq:pvalue-alternative} for the case $\lambda \in \Lambda^{+}$. 

Finally, the construction for $\lambda\in \Lambda^{-}\equiv \{\lambda \in \mathbf R: \lambda < \lambda_0\}$ parallels the one for $\lambda \in \Lambda^{+}$ so we omit the arguments here. Putting all the cases together, \eqref{eq:pvalue-alternative} follows and this completes the proof.  
\end{proof}

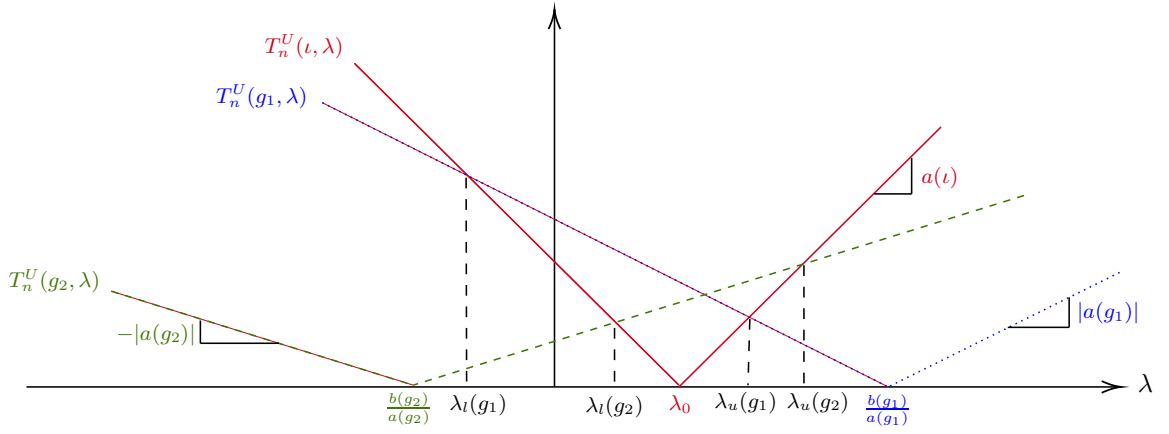
\begin{figure}[t!]
\centering
\scalebox{0.9}{
\tikzset{every picture/.style={line width=0.65pt}} 

\begin{tikzpicture}[x=0.75pt,y=0.75pt,yscale=-1,xscale=1]

\draw    (22.29,219) -- (627.29,219) ;
\draw [shift={(629.29,219)}, rotate = 180] [color={rgb, 255:red, 0; green, 0; blue, 0 }  ][line width=0.75]    (10.93,-3.29) .. controls (6.95,-1.4) and (3.31,-0.3) .. (0,0) .. controls (3.31,0.3) and (6.95,1.4) .. (10.93,3.29)   ;
\draw    (315,219) -- (315,10.02) ;
\draw [shift={(315,8.02)}, rotate = 450] [color={rgb, 255:red, 0; green, 0; blue, 0 }  ][line width=0.75]    (10.93,-3.29) .. controls (6.95,-1.4) and (3.31,-0.3) .. (0,0) .. controls (3.31,0.3) and (6.95,1.4) .. (10.93,3.29)   ;
\draw [color={rgb, 255:red, 208; green, 2; blue, 27 }  ,draw opacity=1][fill={rgb, 255:red, 208; green, 2; blue, 27 }  ,fill opacity=1 ]   (204.11,38.11) -- (384,218) ;
\draw [color={rgb, 255:red, 208; green, 2; blue, 27 }  ,draw opacity=1]   (529.29,73.71) -- (384,219) ;
\draw    (513,91) -- (513,111.11) ;
\draw    (492.29,111.11) -- (513,111.11) ;
\draw [color={rgb, 255:red, 0; green, 0; blue, 255 }  ,draw opacity=1 ][fill={rgb, 255:red, 208; green, 2; blue, 27 }  ,fill opacity=1, dotted ]   (186.29,60.11) -- (500,219) ;
\draw [color={rgb, 255:red, 0; green, 0; blue, 255 }  ,draw opacity=1, dotted ]   (628.29,155.02) -- (500,219) ;
\draw    (600.32,168.91) -- (600.32,185.7) ;
\draw    (566.29,185.7) -- (600.32,185.7) ;
\draw [color={rgb, 255:red, 65; green, 117; blue, 5 }  ,draw opacity=1 ][fill={rgb, 255:red, 208; green, 2; blue, 27 }  ,fill opacity=1, dashed]   (69.29,165.47) -- (236.7,218) ;
\draw [color={rgb, 255:red, 65; green, 117; blue, 5 }  ,draw opacity=1,dashed ]   (575.29,112.02) -- (236.7,218) ;
\draw    (118.69,181.85) -- (118.69,194.6) ;
\draw    (118.69,194.6) -- (162.29,194.6) ;
\draw  [dash pattern={on 4.5pt off 4.5pt}]  (266,102.02) -- (266.29,219.02) ;
\draw  [dash pattern={on 4.5pt off 4.5pt}]  (348.29,186.02) -- (348.29,221.11) ;
\draw  [dash pattern={on 4.5pt off 4.5pt}]  (453.29,151.02) -- (453.29,219.02) ;
\draw  [dash pattern={on 4.5pt off 4.5pt}]  (423.29,181.02) -- (422.29,218.02)  ;

\draw (637,210.4) node [anchor=north west][inner sep=0.75pt]    {$\lambda $};
\draw (377,222.4) node [anchor=north west][inner sep=0.75pt]  [font=\footnotesize,color={rgb, 255:red, 208; green, 2; blue, 27 }  ,opacity=1 ]  {$\lambda _{0} \ $};
\draw (517,94.4) node [anchor=north west][inner sep=0.75pt]  [font=\footnotesize,color={rgb, 255:red, 208; green, 2; blue, 27 }  ,opacity=1 ]  {$a( \iota )$};
\draw (485,221.4) node [anchor=north west][inner sep=0.75pt]  [font=\footnotesize,color={rgb, 255:red, 0; green, 0; blue, 255 }  ,opacity=1 ]  {$\frac{b( g_{1})}{a( g_{1})}$};
\draw (603.32,170.31) node [anchor=north west][inner sep=0.75pt]  [font=\footnotesize,color={rgb, 255:red, 0; green, 0; blue, 255 }  ,opacity=1 ]  {$|a( g_{1}) |$};
\draw (219.7,220.4) node [anchor=north west][inner sep=0.75pt]  [font=\footnotesize,color={rgb, 255:red, 65; green, 117; blue, 5 }  ,opacity=1 ]  {$\frac{b( g_{2})}{a( g_{2})}$};
\draw (70.68,182.14) node [anchor=north west][inner sep=0.75pt]  [font=\footnotesize,color={rgb, 255:red, 65; green, 117; blue, 5 }  ,opacity=1 ]  {$-|a( g_{2}) |$};
\draw (402.29,220.42) node [anchor=north west][inner sep=0.75pt]  [font=\footnotesize]  {$\lambda_u( g_{1})$};
\draw (442.29,220.42) node [anchor=north west][inner sep=0.75pt]  [font=\footnotesize]  {$\lambda_u( g_{2})$};
\draw (126,48.4) node [anchor=north west][inner sep=0.75pt]  [font=\footnotesize,color={rgb, 255:red, 0; green, 0; blue, 255 }  ,opacity=1 ]  {$T_{n}^{U}( g_{1} ,\lambda )$};
\draw (11,151.4) node [anchor=north west][inner sep=0.75pt]  [font=\footnotesize,color={rgb, 255:red, 65; green, 117; blue, 5 }  ,opacity=1 ]  {$T_{n}^{U}( g_{2} ,\lambda )$};
\draw (153,20.4) node [anchor=north west][inner sep=0.75pt]  [font=\footnotesize,color={rgb, 255:red, 208; green, 2; blue, 27 }  ,opacity=1 ]  {$T_{n}^{U}( \iota ,\lambda )$};
\draw (330.79,222.4) node [anchor=north west][inner sep=0.75pt]  [font=\footnotesize]  {$\lambda_l( g_{2})$};
\draw (255.29,221.37) node [anchor=north west][inner sep=0.75pt]  [font=\footnotesize]  {$\lambda_l( g_{1})$};

\end{tikzpicture}}
\caption{$T_n(g, \lambda)$ as functions of $\lambda$ for $g\in \{\iota,g_1,g_2\}$.}
\label{figure:intersections}
\end{figure}

Figure \ref{figure:pvalue} illustrates the $p$-value in \eqref{eq:pvalue-alternative} as a function of $\lambda$ for the groups in Figure \ref{figure:intersections}. Since $ \hat p_n(\lambda)$ is right continuous and increasing for $\lambda<\lambda_0$, we can define $\lambda_l$ as the smallest value of $\lambda$ for which $ \hat p_n(\lambda)\ge \alpha$. Such value exists and is unique. Similar, since $ \hat p_n(\lambda)$ is left continuous and decreasing for $\lambda>\lambda_0$, we can define $\lambda_u$ as the largest value of $\lambda$ for which $ \hat p_n(\lambda)\ge \alpha$. Such value exists and is again unique. This argument leads to the representation of $C_n$ in \eqref{eq:Cn}, showing that ART-based confidence intervals for $c'\beta$ are indeed intervals in $\mathbf R$. Furthermore, note that \eqref{eq:pvalue-alternative} implies that the smallest value of $\lambda$ for which $ \hat p_n(\lambda)\ge \alpha$ can be defined as
\begin{equation*}
    \inf \left \{ \lambda \in \mathbf R:  \frac{1}{|\mathbf{G}|}\sum_{g \in \mathbf{G}} I \left\{\lambda\ge \lambda_{l}(g)\right\}\geq \alpha \right \}~,
\end{equation*}
which is just the definition of the $\alpha$ quantile of $\lambda_l(g)$, as defined in Algorithm \ref{algo:crs3}. A similar result holds for $\lambda_u$ and so $C_n$ can be computed in closed form by Algorithm \ref{algo:crs3}.

\begin{figure}[h!]
\centering
\tikzset{every picture/.style={line width=0.75pt}} 

\begin{tikzpicture}[x=0.75pt,y=0.75pt,yscale=-1,xscale=1]

\draw    (174.29,219) -- (542.29,219) ;
\draw [shift={(544.29,219)}, rotate = 180] [color={rgb, 255:red, 0; green, 0; blue, 0 }  ][line width=0.75]    (10.93,-3.29) .. controls (6.95,-1.4) and (3.31,-0.3) .. (0,0) .. controls (3.31,0.3) and (6.95,1.4) .. (10.93,3.29)   ;
\draw    (315,219) -- (315,10.02) ;
\draw [shift={(315,8.02)}, rotate = 450] [color={rgb, 255:red, 0; green, 0; blue, 0 }  ][line width=0.75]    (10.93,-3.29) .. controls (6.95,-1.4) and (3.31,-0.3) .. (0,0) .. controls (3.31,0.3) and (6.95,1.4) .. (10.93,3.29)   ;
\draw  [dash pattern={on 4.5pt off 4.5pt}]  (266,87.02) -- (266.29,219.02) ;
\draw  [dash pattern={on 4.5pt off 4.5pt}]  (348.29,43.24) -- (348.29,221.11) ;
\draw  [dash pattern={on 4.5pt off 4.5pt}]  (453.29,90.24) -- (453.29,219.02) ;
\draw  [dash pattern={on 4.5pt off 4.5pt}]  (422.29,43.24) -- (422.29,218.02) ;
\draw    (348.29,40.24) -- (422.29,40.24) ;
\draw [shift={(422.29,40.24)}, rotate = 0] [color={rgb, 255:red, 0; green, 0; blue, 0 }  ][fill={rgb, 255:red, 0; green, 0; blue, 0 }  ][line width=0.75]      (0, 0) circle [x radius= 3.35, y radius= 3.35]   ;
\draw [shift={(348.29,40.24)}, rotate = 0] [color={rgb, 255:red, 0; green, 0; blue, 0 }  ][fill={rgb, 255:red, 0; green, 0; blue, 0 }  ][line width=0.75]      (0, 0) circle [x radius= 3.35, y radius= 3.35]   ;
\draw    (266,91.02) -- (345.94,91.02) ;
\draw [shift={(348.29,91.02)}, rotate = 0] [color={rgb, 255:red, 0; green, 0; blue, 0 }  ][line width=0.75]      (0, 0) circle [x radius= 3.35, y radius= 3.35]   ;
\draw [shift={(266,91.02)}, rotate = 0] [color={rgb, 255:red, 0; green, 0; blue, 0 }  ][fill={rgb, 255:red, 0; green, 0; blue, 0 }  ][line width=0.75]      (0, 0) circle [x radius= 3.35, y radius= 3.35]   ;
\draw    (425.64,90.24) -- (453.29,90.24) ;
\draw [shift={(453.29,90.24)}, rotate = 0] [color={rgb, 255:red, 0; green, 0; blue, 0 }  ][fill={rgb, 255:red, 0; green, 0; blue, 0 }  ][line width=0.75]      (0, 0) circle [x radius= 3.35, y radius= 3.35]   ;
\draw [shift={(423.29,90.24)}, rotate = 0] [color={rgb, 255:red, 0; green, 0; blue, 0 }  ][line width=0.75]      (0, 0) circle [x radius= 3.35, y radius= 3.35]   ;
\draw    (455.64,140.13) -- (518.29,140.13) ;
\draw [shift={(453.29,140.13)}, rotate = 0] [color={rgb, 255:red, 0; green, 0; blue, 0 }  ][line width=0.75]      (0, 0) circle [x radius= 3.35, y radius= 3.35]   ;
\draw    (201.29,140.02) -- (263.34,140.02) ;
\draw [shift={(265.69,140.02)}, rotate = 0] [color={rgb, 255:red, 0; green, 0; blue, 0 }  ][line width=0.75]      (0, 0) circle [x radius= 3.35, y radius= 3.35]   ;
\draw  [dash pattern={on 4.5pt off 4.5pt}]  (385.29,40.24) -- (385.29,218.11) ;
\draw    (310.61,39.87) -- (320.04,39.87) ;

\draw (554,212.4) node [anchor=north west][inner sep=0.75pt]    {$\lambda $};
\draw (377,222.4) node [anchor=north west][inner sep=0.75pt]  [font=\footnotesize,color={rgb, 255:red, 0; green, 0; blue, 0 }  ,opacity=1 ]  {$\lambda _{0} \ $};
\draw (402.29,220.42) node [anchor=north west][inner sep=0.75pt]  [font=\footnotesize]  {$\lambda_u( g_{1})$};
\draw (442.29,220.42) node [anchor=north west][inner sep=0.75pt]  [font=\footnotesize]  {$\lambda_u( g_{2})$};
\draw (330.79,222.4) node [anchor=north west][inner sep=0.75pt]  [font=\footnotesize]  {$\lambda_l( g_{2})$};
\draw (255.29,221.37) node [anchor=north west][inner sep=0.75pt]  [font=\footnotesize]  {$\lambda_l( g_{1})$};
\draw (273.35,3.91) node [anchor=north west][inner sep=0.75pt]  [font=\footnotesize]  {$\hat{p}_{n}( \lambda )$};
\draw (301.52,34.01) node [anchor=north west][inner sep=0.75pt]  [font=\footnotesize]  {$1$};

\end{tikzpicture}
\caption{$\hat p_n(\lambda)$ as a function of $\lambda$.}
\label{figure:pvalue}
\end{figure}
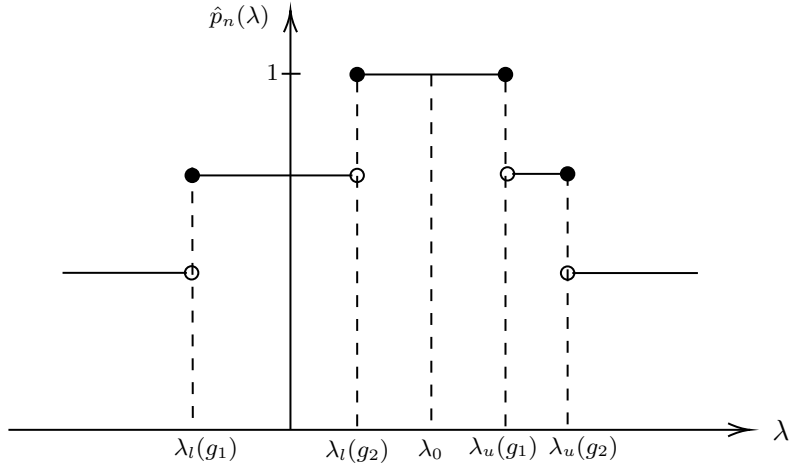


\section{What we need for ARTs to work}\label{sec:requirements}
The main requirement underlying ARTs is Assumption 3.1 in \cite{canay/etal:17}. This assumption guarantees that the test delivers rejection probabilities under the null hypothesis that are close to the nominal level $\alpha$ in an asymptotic framework where $n\to \infty$ and $q$ remains fixed. In the context of the linear model in \eqref{eq:main}, this translates into the following two conditions summarized in Assumption \ref{ass:CRS3.1} below. 

\begin{assumption}\label{ass:CRS3.1}
Let $\{\hat\beta_{n,j}:j\in J\}$ be the cluster-by-cluster estimators of $\beta$ defined in Algorithm \ref{algo:crs1}. Assume that:
\begin{enumerate}[(a)]
    \item $\{\hat\beta_{n,j}:j\in J\}$ jointly converge in distribution at some (possibly unknown) rate; i.e., 
\begin{equation}\label{eq:assump-1}
    \left(
    \begin{array}{c}
    a_{n,1}(\hat\beta_{n,1}-\beta)\\
    \vdots\\
    a_{n,q}(\hat\beta_{n,q}-\beta)\\
    \end{array} \right)\overset{d}{\to} 
    \left(
    \begin{array}{c}
    S_1\\
    \vdots\\
    S_q\\
    \end{array}\right)
\end{equation}
for a sequences $a_{n,j}\to \infty$ and random variables $(S_1,\dots,S_q)'$.
\item The limiting random variables $(S_1,\dots,S_q)'$ are invariant to sign changes, i.e., 
\begin{equation}\label{eq:assump-2}
    (g_1S_1,\dots,g_qS_q) \overset{d}{=} (S_1,\dots,S_q)~,
\end{equation}
for any $g$ in $\mathbf G$, where $\mathbf G$ is defined in Step 4 of Algorithm \ref{algo:crs1}. 
\end{enumerate}

\end{assumption}

Condition \eqref{eq:assump-1} holds, for example, when $Z_{i,j}$ and $\epsilon_{i,j}$ are uncorrelated and the analyst assumes some form of weak dependence within clusters that permits the application of an appropriate central limit theorem. In such a case, \eqref{eq:assump-1} typically holds with $a_{n,j}=\sqrt{n_j}$ and each $S_j$ being a mean-zero normal random variable. In fact, under the commonly used assumption of independent clusters, it also follows that $S_j \indep S_{j'}$ for any $j\ne j'$. In this case the normally distributed random variables may not be identically distributed but are indeed independent. Condition \eqref{eq:assump-2}, in turn, requires each $S_j$ to be symmetrically distributed around zero and independent of each other. This is immediately satisfied when each $S_j$ is a mean-zero normal random variable and clusters are independent. Importantly, these assumptions allow for the normally distributed random variables to have different variances across clusters; a type of heterogeneity not allowed by the cluster wild bootstrap approach popularized by \cite{cameron2008bootstrap} and later studied formally by \cite{canay/santos/shaikh:20}.  

\begin{remark}\label{rem:on-normality}
\rm The asymptotic normality in \eqref{eq:assump-1} arises frequently in applications, but is not necessary for the validity of ARTs. All that is required is that the estimators $\{a_{n,j}(\hat\beta_{n,j}-\beta):j\in J\}$ have a limiting distribution that is the product of $q$ distributions that are symmetric about zero. This may even hold in cases where the estimators have infinite variances or are inconsistent. See \citet[][Remark 4.5]{canay/etal:17} for additional discussion on this point. \qed
\end{remark}

\begin{remark}\label{rem:large-clusters}
\rm It is worthwhile to contrast the requirements of Assumption \ref{ass:CRS3.1} with those of ``classical'' methods, such as those described in \cite{liang1986longitudinal}.  These latter methods permit arbitrary dependence within each cluster, but require the size of the clusters to be small and the number of clusters to be large.  As described above, Assumption \ref{ass:CRS3.1}(a), on the other hand, permits the number of clusters to be small, but requires the size of the clusters to be large and weak dependence within each cluster.  We emphasize, however, that these restrictions are commonly employed in establishing the validity of other methods in settings with a small number of clusters, including, for example, the $t$-test approach by \cite{ibragimov2010t} and the wild bootstrap \cite{canay/santos/shaikh:20}. \qed
\end{remark}




\begin{remark}\label{rem:IV}
    \rm We focus our exposition on the case where $Z_{i,j}$ is exogenous but we emphasize that the conditions in \eqref{eq:assump-1} and \eqref{eq:assump-2} typically hold in instrumental variable (IV) models. Accommodating IV to ARTs then only requires modifying Step 1 in Algorithm \ref{algo:crs1} so that the least squares regression is replaced with the appropriate IV regression. Steps 2-6 remain unaffected. \qed
\end{remark}

An implicit requirement behind ARTs that deserves further comments lies in Step 1 of Algorithm \ref{algo:crs1}, which requires that the analyst runs cluster-by-cluster regressions. This step implicitly assumes that the parameter $\beta$ is identified within each cluster. In practice, this means that the matrix $\hat \Omega_{n,j}$ in \eqref{eq:Omega-hat-j} must be invertible for each $j\in J$ and hence the same requirement applies to Algorithm \ref{algo:crs2}. This restriction may be substantially important in some applications and so here we discuss common ways in which the problem may manifest and two alternative remedies. 

One case in which running least squares cluster-by-cluster is not feasible is when the coefficient of interest is associated with a variable that only varies across clusters. For example, consider the model in \eqref{eq:main} and partition $Z_{i,j}$ into a constant term, a scalar variable that only varies across clusters, $Z^{(1)}_j$, and another variable that varies across and within clusters, $Z^{(2)}_{i,j}$. That is, 
\begin{equation}\label{eq:reg-issue-1}
    Y_{i,j} = \beta_0 + Z^{(1)}_{j}\beta_1 + Z^{(2)}_{i,j}\beta_2+ \epsilon_{i,j}~,
\end{equation}
where the analysts' interest lies in the coefficient $\beta_1$, i.e., $c'\beta = \beta_1$. Clearly, the regression in Step 1 of Algorithm \ref{algo:crs1} would not separately identify $\beta_0$ and $\beta_1$ as $Z^{(1)}_{j}$ is perfectly colinear with the constant term. The matrix $\hat \Omega_{n,j}$ in \eqref{eq:Omega-hat-j} is simply singular. This situation arises, for example, in the empirical application considered by \cite{canay/etal:17supp} where $j\in J$ indexes schools and the variable of interest is a treatment indicator at the school level. A natural remedy in a situation like this is clustering more coarsely (e.g., by combining clusters) to obtain variation within the re-defined clusters. This is possible for ARTs since the validity of the method does not rely on having a large number of clusters and thus it can afford to work with coarser clustering. In fact, in certain settings combining clusters may be quite natural. For example, \cite{canay/etal:17supp} re-defined clusters as ``pairs'' of schools (as opposed to just schools) given that the treatment assignment mechanism of the experiment was a matched pairs design and so the pairs used at the randomization stage represented natural groupings. In other settings where it is less clear how to group clusters, any grouping that satisfies the requisite identification condition leads to a valid test, but it may be further desirable to combine such tests to limit concerns about ``data snooping'' across groupings.  To this end, results in \cite{diciccio2020exact} on combining tests may be relevant.

\begin{remark}\label{rem:what-not-to-do}
\rm A quick inspection of \eqref{eq:reg-issue-1} may lead the analyst to believe there is a workaround that does not involve combining clusters if one instead uses some estimator of $\beta_0$ from a full sample regression. For example, the full sample least squares estimator $\hat\beta_{n,0}$ from the regression in \eqref{eq:main}. Then, assuming for simplicity that $Z_j^{(1)}\ne 0$ for all $j\in J$, one may consider modifying Step 1 in Algorithm \ref{algo:crs1} by running a regression of $Y_{i,j}$ on an intercept and $Z^{(2)}_{i,j}$ (not including $Z_j^{(1)}$) and then redefining $\{\hat \beta_{n,j}:j \in J\}$ as the difference between the within cluster intercept estimates, $\hat\beta_{j,0}$ and the full sample estimate $\hat\beta_{n,0}$, i.e., $\hat \beta_{n,j}=\hat\beta_{j,0}-\hat\beta_{n,0}$. Such strategies unfortunately introduce dependence between the $q$ estimators of $\beta$ (as they all depend on $\hat\beta_{n,0}$) and thus end up violating one of the two main conditions needed for ARTs to be asymptotically valid; mainly condition \eqref{eq:assump-2}. \qed
\end{remark}

Another case where the lack of identification within cluster may manifest is when the variable of interest actually varies within clusters but the model specification involves other variables that are collinear with some other variable (including the variable of interest or the constant term) within clusters. For example, consider the model in \eqref{eq:main} where instead of individuals indexed by $i\in I_{n,j}$, units within cluster are indexed over time $t\in T$. Partition $Z_{j,t}$ into the variable of interest, $Z^{(1)}_{j,t}$, and time fixed effects $\delta_t$. That is, 
\begin{equation}\label{eq:reg-issue-2}
    Y_{j,t} = Z^{(1)}_{j,t}\beta_1 + \sum_{\tilde t\in T}I\{\tilde t=t\}\delta_{\tilde t}+ \epsilon_{j,t}~.
\end{equation}
It then follows that, within each cluster $j\in J$, the time fixed effect $\delta_t$ absorbs all the variation in $Z^{(1)}_{j,t}$ and so $\beta_1$ is not identified. In cases like this the analyst could again combine clusters to obtain variation within the re-defined clusters. An alternative remedy is to change the specification by, for example, replacing the time fixed effect with a cluster-specific time trend. Such specification is more restrictive than the time fixed effect in the sense that it imposes a linear trend but, at the same time, is more general as it allows for heterogeneity across clusters in the linear trend. We illustrate this approach in the application we consider in Section \ref{sec:MQY}. 

The need to identify $\beta$ within each cluster is in our view the main limitation of ARTs, but a limitation that needs to be dealt with in certain settings. One may then wonder why not simply use some other inference method that is valid when the number of clusters is small and that does not rely on estimating $\beta$ cluster-by-cluster. Perhaps the most popular approach in that category is the cluster wild bootstrap popularized by \cite{cameron2008bootstrap} and recently studied formally by \cite{canay/santos/shaikh:20}. While not having to estimate $\beta$ within each cluster represents an advantage over ARTs, this additional flexibility comes at a cost in terms of the degree of heterogeneity that the model can deal with. In particular, the results in  \cite{canay/santos/shaikh:20} show that the cluster wild bootstrap is expected to work well in settings with a small number of clusters \emph{as long as} the clusters are ``homogeneous,'' in a  sense made precise in \cite{canay/santos/shaikh:20}. Intuitively, it is required that the variance covariance matrix $\hat \Omega_{n,j}$ defined in \eqref{eq:Omega-hat-j} is the same across clusters (up to scalar multiplication). Such stringent homogeneity condition is not required for ARTs to work well, as the method allows clusters to be arbitrarily heterogeneous as long as $\hat \Omega_{n,j}$ is invertible for $j\in J$. 

\begin{remark}\label{rem:c-prime-beta}
\rm For ease of exposition, we have written the requirement in \eqref{eq:assump-1} in terms of the differences $\hat \beta_{n,j} - \beta$, but it is possible to replace it with the differences $c’\hat \beta_{n,j} - c’\beta$ (or $R\hat \beta_{n,j} - R\beta$, depending on the null hypotheses of interest).  In most cases, re-writing the condition in this way is not useful, but it is in cases where $c’\beta$ is identified within each cluster while $\beta$ is not. For example, consider the model in \eqref{eq:reg-issue-1} when the coefficient of interest is $\beta_2$ as opposed to $\beta_1$, i.e, $c'\beta = \beta_2$. In that case the entire term $\beta_0+Z_j^{(1)}\beta_1$ may be absorbed into a cluster-specific intercept without affecting the identification and estimation of $c'\beta = \beta_2$ within each cluster. \qed
\end{remark}

\section{Empirical applications}\label{sec:applications}

In this section we apply ARTs as described in Algorithm \ref{algo:crs1} and ART-based confidence intervals as described in Algorithm \ref{algo:crs3} in the context of two distinct empirical applications. The \texttt{R} and \texttt{Stata} packages and codes required to replicate the results in this section are available as part of the online supplemental material.  

\subsection{Meng, Qian and Yared (2015)}\label{sec:MQY}
\citet[][MQY]{mengQY2015} argue that China's Great Famine, from 1959 to 1961, was the result of an inflexible food procurement policy by the central government. To make this point, they show that food production and mortality become positively correlated during the time of famine, when this coefficient is otherwise negative or not significantly different from 0 in normal times. 

MQY consider the following regression,
	\begin{equation*}
	Y_{j,t+1} = Z^{(1)}_{j,t}\beta_1 + Z^{(2)}_{j,t}\beta_2 + \delta_t + \epsilon_{j,t}
	\end{equation*}
where $j$ indexes provinces (ranging from 1 to 19) and $t$ indexes years (ranging from 1953 to 1982). Here, 
\begin{align*}
	Y_{j, t+1} & = \log(\text{number of deaths in province $j$ during year $t+1$}) \\
	Z^{(1)}_{j,t} & = \log(\text{predicted grain production in province $j$ during year $t$}) \\ & \qquad \times I\{\text{$t$ is a famine year} \} \\
	Z^{(2)}_{j,t} & = \log(\text{predicted grain production in province $j$ during year $t$}) \\
	\delta_t & = \text{time fixed effects}	~.
\end{align*}
In this application the level of clustering is a province, and so in order to apply ARTs as described in Section \ref{sec:algorithm}, one needs to estimate $\beta=(\beta_1,\beta_2)'$ and $\delta_t$ province-by-province. This illustrates one of the situations where including time fixed effects province-by-province is infeasible for the implementation of ARTs, given that the only source of remaining variation within a province is indeed time. The second identification problem described in Section \ref{sec:requirements} then arises. As we discussed in that section, one way to deal with this issue consists of replacing the time fixed effects with a cluster-specific time trend, i.e., in Step 1 of Algorithm \ref{algo:crs1} estimate
	\begin{equation}\label{eq:mqy-regression}
	Y_{j,t+1} = Z^{(1)}_{j,t}\beta_1 + Z^{(2)}_{j,t}\beta_2 + \gamma_j t + \epsilon_{j,t}~.
	\end{equation}
We will refer to this as Analysis \#1. In addition, we also consider the following alternative specifications studied by MQY:
\begin{itemize}
	\item Analysis \#2: Repeating Analysis \#1 using only data between 1953 and 1965.
	\item Analysis \#3: Repeating Analysis \#1 using four additional autonomous provinces.
	\item Analysis \#4: Repeating Analysis \#2 using four additional autonomous provinces.
	\item Analysis \#5: Repeating Analysis \#1 using actual rather than constructed grain production. 
	\item Analysis \#6: Repeating Analysis \#2 using actual rather than constructed grain production. 
\end{itemize}
As with Analysis \#1, the above analyses differ from their MQY counterparts only in that a linear time trend $\gamma_j t$ replaces time fixed effects $\delta_t$. Table \ref{tab:mqy_summary} summarizes the number of clusters and the number of observations for each of these analyses.  We caution, however, that in this application, in addition to the number of clusters being small, the number of observations within each cluster may also be small.  See Remark \ref{rem:large-clusters} for further discussion in relation to Assumption \ref{ass:CRS3.1}(a).


\begin{table}[htbp]
  	\begin{center}
     \scalebox{0.77}{
    \begin{tabular}{cccccc}
    \toprule
    Analysis & \# of Clusters & Min. Size & Med. Size & Max. Size & Mean \\
    \midrule
    \#1, \#5& 19  & 29  & 30  & 30  & 29.95  \\
    \#2, \#6& 19  & 12  & 13  & 13  & 12.95  \\
    \#3  	& 23  & 29  & 30  & 30  & 29.96  \\
    \#4  	& 23  & 12  & 13  & 13  & 12.96  \\
    \bottomrule
    \end{tabular}
    }
    \end{center}
    \caption{\footnotesize Cluster Information. `Min. Size', `Med. Size', `Max. Size' denote the minimum, the median, and the maximum size of clusters.}
  \label{tab:mqy_summary}%
\end{table}

\cite{mengQY2015} consider the following two null hypotheses of interest, 
\begin{equation}\label{eq:mqy-nulls}
	H_0^{(1)}: \beta_1 = 0 \quad \text{ and } \quad H_0^{(2)}: \beta_1 + \beta_2 = 0~.
\end{equation}
In Table \ref{tab:mqy_results} we replicate the main table in \cite{mengQY2015} using cluster robust standard errors (CCE) and also include the results associated with ARTs for both $H_0^{(1)}$ and $H_0^{(2)}$ in \eqref{eq:mqy-nulls}. For $H_0^{(1)}$ we report $p$-values and 95\% confidence intervals, while for $H_0^{(2)}$ we just report $p$-values following MQY. The authors note in footnote 33 that using the cluster wild bootstrap led to similar results as those presented in their main table so we do not include cluster wild bootstrap results here either. 

\begin{table}[t]
  \begin{center}
  \scalebox{0.77}{
   \begin{tabular}{ccccccc}
    \toprule
          & \#1  & \#2  & \#3  & \#4  & \#5  & \#6 \\
    \midrule
    \multicolumn{1}{l}{LS Estimate: $\beta_1$} & 0.063 & 0.057 & 0.071 & 0.067 & 0.064 & 0.058 \\
          &       &       &       &       &       &  \\
    \multicolumn{1}{l}{CCE: Province} &       &       &       &       &       &  \\
    \multicolumn{1}{r}{se} & 0.007 & 0.007 & 0.007 & 0.008 & 0.007 & 0.007 \\
    \multicolumn{1}{r}{$p$-value} & 0.000 & 0.000 & 0.000 & 0.000 & 0.000 & 0.000 \\
    \multicolumn{1}{r}{95\% CI} & [0.050, 0.077] & [0.043, 0.071] & [0.057, 0.086] & [0.051, 0.083] & [0.051, 0.078] & [0.044, 0.071] \\
          &       &       &       &       &       &  \\
    \multicolumn{1}{l}{ART} &       &       &       &       &       &  \\
    \multicolumn{1}{r}{$p$-value} & 0.000 & 0.002 & 0.000 & 0.000 & 0.000 & 0.000 \\
    \multicolumn{1}{r}{95\% CI} & [0.032, 0.055] & [0.018, 0.047] & [0.038, 0.066] & [0.028, 0.067] & [0.032, 0.058] & [0.029, 0.050] \\
          &       &       &       &       &       &  \\
    \multicolumn{1}{l}{$\beta_1 + \beta_2 = 0$} &       &       &       &       &       &  \\
    \multicolumn{1}{r}{CCE $p$-value} & 0.050 & 0.009 & 0.059 & 0.005 & 0.266 & 0.363 \\
    \multicolumn{1}{r}{ART $p$-value} & 0.098 & 0.571 & 0.096 & 0.487 & 0.080 & 0.001 \\
    \midrule
    Observations & 569   & 246   & 689   & 298   & 569   & 246 \\
    Short Sample & No    & Yes   & No    & Yes   & No    & Yes \\
    Auto. Region & No    & No    & Yes   & Yes   & No    & No \\
    Pred. Grain Prod. & Yes   & Yes   & Yes   & Yes   & No    & No \\
    \bottomrule
    \end{tabular}%
  }
  \end{center} 
      \caption{\footnotesize Results for Analyses \#1-6, comparable to those in Table 2 of Meng, Qian and Yared (2015). `LS Estimate' denotes the full sample OLS estimate for $\beta_1$. CCE refers to cluster-robust standard errors. ART $p$-values are obtained using Algorithm \ref{algo:crs1}. ART-based 95\% confidence intervals are obtained using Algorithm \ref{algo:crs3}. }
  \label{tab:mqy_results}%
\end{table}%

We comment on the following main features of Table \ref{tab:mqy_results}:
\begin{enumerate}
	\item For the null hypothesis $H_0^{(1)}$ associated with the parameter $\beta_1$, the ART $p$-values are of comparable magnitude to traditional CCE $p$-values. Similarly, ART-based confidence intervals are of roughly the same length as those obtained based on CCE although the ART-based confidence intervals do not contain the LS estimates. This is because ART-based confidence intervals are centered around the mean of the province-by-province estimates, which may not necessarily be equal to the full sample LS estimate of $\beta_1$. 
	\item For the null hypothesis $H_0^{(2)}$ associated with the parameter $\beta_1 + \beta_2$, the ART $p$-value is sometimes higher and sometimes lower than the CCE $p$-value depending on the specification. Given the relatively small number of clusters in this application, the ART $p$-values are likely to be more reliable than those associated with CCE as CCE is known to perform poorly when the number of clusters is not sufficiently large. 
\end{enumerate}

\subsection{Munyo and Rossi (2015)}\label{sec:MR}
\cite{munyo2015first} study criminal recidivism of former prisoners by looking at the relationship between the number of inmates released from incarceration on a given day and the number of offenses committed on the same day. They claim that the liquidity constraints that inmates face on the day of release increase the likelihood of recidivism on the same day. Using data of 2631 days between January 1st 2004 and March 15 2011 collected from the criminal incidents reports in Montevideo in Uruguay, they estimate the following linear model by least squares
\begin{align*}
Y_{t} = Z_{t}'\beta + \epsilon_{t} 
\end{align*}
where $t$ indexes days and 
\begin{align*}
Y_{t} &= \text{the total number of offenses on day }t\\
Z_{t} &= \text{the total number of inmates released, temperature, rainfall, hours of sunshine}\\
 &{} \text{on day $t$, a dummy for holidays, a dummy for December 31st and a yearly trend.}
\end{align*}
We refer to this as Analysis \#1. \cite{munyo2015first} additionally consider the following four analyses:
\begin{itemize}
\item Analysis \#2: $Z_{t}$ includes a daily trend in place of a yearly trend.
\item Analysis \#3: $Z_{t}$ includes a monthly trend in place of a yearly trend.
\item Analysis \#4: $Z_{t}$ includes an intra-month daily trend, month- and year- level fixed effects and their interactions in place of a yearly trend.
\item Analysis \#5: $Z_{t}$ includes month- and year- level fixed effects and their interactions in place of a yearly trend.
\end{itemize}
Analysis \#5 is their preferred specification. \cite{munyo2015first} report the results of these analyses in Table 2 in their paper. They report least squares estimates of $\beta$ with Newey-West heteroskedasticity-autocorrelation-consistent (HAC) standard errors. In addition, they report ART $p$-values as described in Algorithm \ref{algo:crs1} for the null hypothesis that $H_0: c'\beta=0$ as in \eqref{eq:H0-1}, where $c$ selects the coefficient on the total number of inmates released on day $t$. 

In this application the level of clustering is not naturally determined by the data, but pseudo-clusters may be formed using blocks of consecutive observations under the assumption of weak temporal dependence. In order to apply ARTs as described in Algorithm \ref{algo:crs1} we then form $q$ pseudo-clusters by dividing the data into $q$ consecutive blocks of size $b_n = \lfloor n/q \rfloor$ where $n=2631$ is the number of total observations. More concretely, we define the $j$th pseudo-cluster as
\begin{align*}
	X^{(n)}_j = \{(Y_t, Z'_t)': t=(j-1)b_n + 1, \cdots, jb_n\} \quad \text{where}\quad j=1, \cdots, q-1~,
\end{align*}
and let the last $q$th pseudo-cluster contain all the remaining $n- b_n(q-1)$ observations. Note that in this application the number of pseudo-clusters $q$ is a tuning parameter that the analyst must specify. \cite{munyo2015first} set $q=10$. We repeat their analyses with alternative values of $q$ and investigate how sensitive the results are to this choice. The relevant cluster information is given in Table \ref{tab:MRclustersize}.

\begin{table}[htbp]
  \centering
    \begin{tabular}{cc}
    \toprule
    \# of Clusters (q) & Cluster Size \\ \hline
    8     & 328 \\
    10    & 263 \\
    16    & 164 \\
    \bottomrule
    \end{tabular}%
   \caption{\footnotesize Pseudo-cluster size for different values of $q$.}
\label{tab:MRclustersize}%
\end{table}

\begin{table}
  \centering
  \footnotesize
    \begin{tabular}{lccccc}
    \toprule
    \multicolumn{1}{c}{Specification} & \#1   & \#2   & \#3   & \#4   & \#5 \\
    \midrule
    LS Estimate & 0.225 & 0.260 & 0.259 & 0.225 & 0.234 \\
          &       &       &       &       &  \\
    HAC   &       &       &       &       &  \\
    \multicolumn{1}{r}{se} & 0.124 & 0.123 & 0.123 & 0.096 & 0.096 \\
    \multicolumn{1}{r}{$p$-value} & 0.068 & 0.034 & 0.034 & 0.019 & 0.015 \\
    \multicolumn{1}{r}{95\% CI} & [-0.017, 0.468] & [0.02, 0.5] & [0.019, 0.5] & [0.038, 0.413] & [0.046, 0.421] \\
    ART: q=8 &       &       &       &       &  \\
    \multicolumn{1}{r}{$p$-value} & 0.008 & 0.023 & 0.023 & 0.102 & 0.102 \\
    \multicolumn{1}{r}{95\% CI} & [0.124, 0.429] & [0.035, 0.391] & [0.035, 0.391] & [-0.07, 0.397] & [-0.067, 0.418] \\
    ART: q=10 &       &       &       &       &  \\
    \multicolumn{1}{r}{$p$-value} & 0.002 & 0.014 & 0.014 & 0.063 & 0.053 \\
    \multicolumn{1}{r}{95\% CI} & [0.141, 0.603] & [0.068, 0.446] & [0.068, 0.458] & [-0.023, 0.431] & [-0.003, 0.452] \\
    ART: q=16 &       &       &       &       &   \\
    \multicolumn{1}{r}{$p$-value} & 0.002 & 0.006 & 0.006 & 0.027 & 0.010 \\
    \multicolumn{1}{r}{95\% CI} & [0.131, 0.444] & [0.097, 0.369] & [0.087, 0.371] & [0.02, 0.324] & [0.056, 0.367] \\
          &       &       &       &       &  \\
    \midrule
    Observations & 2631  & 2631  & 2631  & 2631  & 2631 \\
    Time Trend & Year  & Day   & Month & Intra-month Day & None \\
    Time Fixed Effect & No    & No    & No    & Yes   & Yes \\
    Controls & No    & No    & No    & No    & No \\
    \bottomrule
    \end{tabular}%

  \caption{\footnotesize Results for Analyses \#1-5, comparable to those in Table 2 of \cite{munyo2015first}. `LS Estimate' denotes the full sample LS estimate of $\beta$. HAC refers to the heteroskedasticity and autocorrelation consistent standard error. ART $p$-values are obtained using Algorithm \ref{algo:crs1}. ART-based 95\% confidence intervals are obtained using Algorithm \ref{algo:crs3}. }
  \label{tab:munyorossiART}%
\end{table}

Table \ref{tab:munyorossiART} shows LS estimates of $\beta$, $p$-values for the hypothesis in \eqref{eq:H0-1}, and $95\%$ confidence intervals for each analysis. Following \cite{munyo2015first}, we report results based on HAC standard errors. The table also shows ART $p$-values as described in Algorithm \ref{algo:crs1} and ART-based 95\% confidence intervals as described in Algorithm \ref{algo:crs3} for $q=8$, $q=10$, and $q=16$. 

We summarize the main findings of the results in Table \ref{tab:munyorossiART} as follows:
\begin{enumerate}
\item The choice of $q$ is important for the results of ARTs but currently there is no theory developed to choose this tuning parameter according to some data dependent criteria. The smaller $q$ is, the more observations are available within each cluster. Having more observations per cluster is important for one of the requirements behind ARTs, mainly \eqref{eq:assump-1}. A small value of $q$, however, tends to affect the power of ARTs despite not really affecting the control of the rejection probability under the null hypothesis. This feature can be seen in Table \ref{tab:munyorossiART}, where ARTs $p$-values are decreasing in $q$ across different specifications. In this application, where there are still over a hundred observations when $q=16$, a larger value of $q$ like $q=10$ or $q=16$ may be preferable to smaller values, like $q=8$, based on power considerations. Note, however, that except in Analyses \#4--5, where the choice of $q$ determines whether the null hypothesis is rejected at a given significance level, the results for Analyses \#1--3 are in all agreement at a $5\%$ level.

\item Overall, the test results based on standard $t$-test with HAC standard errors are consistent to those of ARTs when $q=16$. Both methods reject the null hypothesis $H_0:c'\beta=0$ at a $10\%$ nominal level across different specifications. The results support the authors' argument that the release of inmates from incarceration increase the chance of re-offenses on the day of release.
\end{enumerate}

\subsection{Computational gains of the new algorithm}\label{sec:gains}
Tables \ref{tab:computational-gains-mqy} and \ref{tab:computational-gains-mr} report four alternative ways to compute ART-based confidence intervals in the two empirical applications we consider in this paper; \cite{mengQY2015} and \cite{munyo2015first}. The first alternative is to compute the confidence intervals by a simple grid search algorithm. The second alternative involves a bi-section algorithm. We implement both of these methods using a studentized and an unstudentized test statistic to illustrate the result in Section \ref{sec:studentized}. The last alternative is to simply use Algorithm \ref{algo:crs3}, as reported in Sections \ref{sec:MQY} and \ref{sec:MR}. In each case, we also report computational times to illustrate the computational advantages of the algorithm we propose in this paper.  The \texttt{R} and \texttt{Stata} codes required to replicate the results in this section are available as part of the online supplemental material.

Starting from Table \ref{tab:computational-gains-mqy}, we see that grid search take a significant amount of time to compute. Our convexity result (Lemma \ref{lem:p-value}) facilitates the use of the bisection method, cutting implementation time by a factor of over 50. Moving from the bisection method to Algorithm \ref{algo:crs3} further leads to a speed up of at least 2 times. A similar pattern emerges in Table \ref{tab:computational-gains-mr}. Furthermore, comparing specification with $q=8$ that with $q=16$, the speed advantage of our method becomes far starker. For $q=16$, grid search takes almost 100 times as long as the bisection method. The bisection method, meanwhile, takes close to 10 times as long as Algorithm \ref{algo:crs3}. 

\begin{table}[htbp]
  \centering \footnotesize
    \begin{tabular}{cccccc}\hline\hline
          & \multicolumn{2}{c}{Grid Search} & \multicolumn{2}{c}{Bisection} & \multirow{2}[1]{*}{ART} \\
          & Stud. & Unstud. & Stud. & Unstud. &  \\
    \midrule
    \multirow{2}[2]{*}{\#1} & [0.032, 0.055] & [0.032, 0.055] & [0.032, 0.055] & [0.032, 0.055] & [0.032, 0.055] \\
          & 19.65 & 6.66  & 0.31  & 0.11  & 0.06 \\
    \midrule
    \multirow{2}[2]{*}{\#2} & [0.018, 0.047] & [0.018, 0.047] & [0.018, 0.047] & [0.018, 0.047] & [0.018, 0.047] \\
          & 46.63 & 16.43 & 0.35  & 0.12  & 0.02 \\
    \midrule
    \multirow{2}[2]{*}{\#3} & [0.038, 0.066] & [0.038, 0.066] & [0.038, 0.066] & [0.038, 0.066] & [0.038, 0.066] \\
          & 24.50 & 8.67  & 0.30  & 0.09  & 0.03 \\
    \midrule
    \multirow{2}[2]{*}{\#4} & [0.028, 0.067] & [0.028, 0.067] & [0.028, 0.067] & [0.028, 0.067] & [0.028, 0.067] \\
          & 62.52 & 21.56 & 0.34  & 0.11  & 0.02 \\
    \midrule
    \multirow{2}[2]{*}{\#5} & [0.032, 0.058] & [0.032, 0.058] & [0.032, 0.058] & [0.032, 0.058] & [0.032, 0.058] \\
          & 19.47 & 6.76  & 0.27  & 0.11  & 0.01 \\
    \midrule
    \multirow{2}[1]{*}{\#6} & [0.029, 0.050] & [0.029, 0.050] & [0.029, 0.050] & [0.029, 0.050] & [0.029, 0.050] \\
          & 23.32 & 8.26  & 0.30  & 0.11  & 0.01 \\ \hline\hline
    \end{tabular}%
   \caption{\footnotesize Computational gains of Algorithm \ref{algo:crs3} relative to grid search and bisection algorithms in the applications of Section \ref{sec:MQY}. The top row for each specification is the confidence interval. The bottom row is time in seconds. For the bisection search, our tolerance is set to the absolute value of the LS estimate, divided by 1000. For comparability, we set the step-size of the grid search to the same value. }
\label{tab:computational-gains-mqy}%
\end{table}%

\begin{table}[htbp]
  \centering \footnotesize
    \begin{tabular}{ccccccc}\hline\hline
          &       & \multicolumn{2}{c}{Grid Search} & \multicolumn{2}{c}{Bisection} & \multirow{2}[1]{*}{ART} \\
          &       & Stud. & Unstud. & Stud. & Unstud. &  \\
    \midrule
    \multirow{10}[10]{*}{$q=8$} & \multirow{2}[2]{*}{\#1} & [0.124, 0.429] & [0.124, 0.429] & [0.124, 0.429] & [0.124, 0.429] & [0.124, 0.429] \\
          &       & 2.62  & 0.92  & 0.11  & 0.02  & 0.06 \\
\cmidrule{2-7}          & \multirow{2}[2]{*}{\#2} & [0.035, 0.391] & [0.035, 0.391] & [0.036, 0.391] & [0.036, 0.391] & [0.035, 0.391] \\
          &       & 3.85  & 1.40  & 0.08  & 0.03  & 0.00 \\
\cmidrule{2-7}          & \multirow{2}[2]{*}{\#3} & [0.035, 0.391] & [0.035, 0.391] & [0.035, 0.390] & [0.035, 0.390] & [0.035, 0.390] \\
          &       & 4.09  & 1.50  & 0.08  & 0.03  & 0.00 \\
\cmidrule{2-7}          & \multirow{2}[2]{*}{\#4} & [-0.070,  0.397] & [-0.070,  0.397] & [-0.070,  0.397] & [-0.070,  0.397] & [-0.070,  0.397] \\
          &       & 9.29  & 3.37  & 0.11  & 0.03  & 0.00 \\
\cmidrule{2-7}          & \multirow{2}[2]{*}{\#5} & [-0.067,  0.418] & [-0.067,  0.418] & [-0.067,  0.418] & [-0.067,  0.418] & [-0.067,  0.418] \\
          &       & 9.20  & 3.18  & 0.07  & 0.04  & 0.01 \\
    \midrule
    \multirow{10}[10]{*}{$q=10$} & \multirow{2}[2]{*}{\#1} & [0.141, 0.603] & [0.141, 0.603] & [0.141, 0.603] & [0.141, 0.603] & [0.141, 0.603] \\
          &       & 30.19 & 10.41 & 0.33  & 0.11  & 0.01 \\
\cmidrule{2-7}          & \multirow{2}[2]{*}{\#2} & [0.068, 0.446] & [0.068, 0.446] & [0.068, 0.446] & [0.068, 0.446] & [0.069, 0.445] \\
          &       & 26.61 & 9.34  & 0.33  & 0.13  & 0.00 \\
\cmidrule{2-7}          & \multirow{2}[2]{*}{\#3} & [0.067, 0.458] & [0.067, 0.458] & [0.068, 0.458] & [0.068, 0.458] & [0.068, 0.458] \\
          &       & 28.37 & 9.78  & 0.33  & 0.11  & 0.02 \\
\cmidrule{2-7}          & \multirow{2}[2]{*}{\#4} & [-0.024,  0.431] & [-0.024,  0.431] & [-0.024,  0.430] & [-0.024,  0.430] & [-0.023,  0.430] \\
          &       & 32.75 & 11.47 & 0.32  & 0.11  & 0.02 \\
\cmidrule{2-7}          & \multirow{2}[2]{*}{\#5} & [-0.003,  0.452] & [-0.003,  0.452] & [-0.003,  0.452] & [-0.003,  0.452] & [-0.003,  0.451] \\
          &       & 31.67 & 11.02 & 0.34  & 0.11  & 0.02 \\
    \midrule
    \multirow{10}[10]{*}{$q=16$} & \multirow{2}[2]{*}{\#1} & [0.124, 0.447] & [0.124, 0.447] & [0.124, 0.447] & [0.124, 0.447] & [0.124, 0.447] \\
          &       & 373.86 & 127.07 & 3.19  & 1.11  & 0.13 \\
\cmidrule{2-7}          & \multirow{2}[2]{*}{\#2} & [0.098, 0.364] & [0.098, 0.364] & [0.098, 0.364] & [0.098, 0.364] & [0.097, 0.364] \\
          &       & 451.67 & 153.20 & 3.16  & 1.11  & 0.14 \\
\cmidrule{2-7}          & \multirow{2}[2]{*}{\#3} & [0.088, 0.368] & [0.088, 0.368] & [0.088, 0.368] & [0.088, 0.368] & [0.087, 0.368] \\
          &       & 415.67 & 142.55 & 3.25  & 1.10  & 0.16 \\
\cmidrule{2-7}          & \multirow{2}[2]{*}{\#4} & [0.014, 0.325] & [0.014, 0.325] & [0.015, 0.325] & [0.015, 0.325] & [0.014, 0.325] \\
          &       & 703.13 & 248.68 & 3.43  & 1.14  & 0.11 \\
\cmidrule{2-7}          & \multirow{2}[2]{*}{\#5} & [0.048, 0.365] & [0.048, 0.365] & [0.048, 0.365] & [0.048, 0.365] & [0.047, 0.365] \\
          &       & 572.54 & 193.69 & 3.09  & 1.10  & 0.14 \\
    \hline\hline
    \end{tabular}%
   \caption{\footnotesize Computational gains of Algorithm \ref{algo:crs3} relative to grid search and bisection algorithms in the applications of Section \ref{sec:MR}. The top row for each specification is the confidence interval. The bottom row is time in seconds. For the bisection search, our tolerance is set to the absolute value of the LS estimate, divided by 1000. For comparability, we set the step-size of the grid search to the same value. }
\label{tab:computational-gains-mr}%
\end{table}%

\section{Concluding remarks}\label{sec:conclude}
The goal of this paper is to make the general theory developed in \cite{canay/etal:17} more accessible by providing a step-by-step algorithmic description of how to implement the test and construct confidence intervals in linear regression models with clustered data, as well as clarifying the main requirements and limitations of the approach. The main two takeaways are the following. First, ARTs-based confidence intervals for scalar parameters in linear regression models can be characterized in closed form and thus are straightforward to implement in practice. Algorithms \ref{algo:crs1} and \ref{algo:crs3} provide a clear explanation of how to apply ARTs in linear models, and the companion {\tt Stata} and {\tt R} packages available as part of the supplemental material are intended to facilitate doing so. Second, our discussion on the main requirements behind ARTs hopefully show that understanding the trade-offs between ARTs and other popular alternatives for inference with a small number of clusters, like the cluster wild bootstrap, is fundamental for practitioners to choose a method that aligns well with the features of their application. In particular, while ARTs essentially demand that the parameter of interest is suitably estimable cluster-by-cluster without imposing restrictions on the degree of heterogeneity across clusters, the cluster wild bootstrap requires the clusters to be sufficiently homogeneous \citep[see][]{canay/santos/shaikh:20} without demanding identification of the parameter of interest cluster-by-cluster. 

\renewcommand{\theequation}{\Alph{section}-\arabic{equation}}

\setcounter{lemma}{0}
\setcounter{theorem}{0}
\setcounter{corollary}{0}
\setcounter{equation}{0}
\setcounter{remark}{0}

\begin{appendices}
\begin{small}	
\end{small}
\end{appendices}

\phantomsection
\bibliography{crs_ref}
\addcontentsline{toc}{section}{References}

\end{document}